\documentclass[11pt]{article}

\usepackage[bottom]{footmisc}
\usepackage{amssymb}
\usepackage{amsmath}
\usepackage{amsthm}
\usepackage{amsfonts,amsthm,amsmath,amssymb}

\usepackage[
    disableredefinitions,
    classfont=bold,
    full,
    langfont=caps,
    funcfont=roman
]{complexity}

\usepackage[mathscr]{euscript}
\usepackage[margin=2cm]{geometry}
\usepackage{bm}
\usepackage[
	backend=bibtex8,
	style=alphabetic,
	natbib=true,
	maxnames=10,
	maxitems=10,
	maxcitenames=3,
	maxalphanames=4,
	minalphanames=3,
	labelalpha=true,
]{biblatex}

\usepackage{dsfont}

\usepackage{hyperref}
\hypersetup{
    citecolor=red,
    colorlinks=true,
    linkcolor=blue,
    filecolor=magenta,      
    urlcolor=cyan,
}

\usepackage{thmtools}
\usepackage{thm-restate}
\usepackage[noabbrev,capitalise]{cleveref}
\usepackage{mathtools}

\declaretheorem[numberwithin=section]{thm}
\declaretheorem[sibling=thm,name=Theorem]{theorem}
\declaretheorem[sibling=thm,name=Lemma]{lemma}
\declaretheorem[sibling=thm,name=Corollary]{corollary}

\declaretheorem[numberwithin=section,name=Definition]{definition}

\usepackage{soul}
\usepackage{xcolor}
\usepackage{wasysym}
\usepackage{comment}
\usepackage{marvosym}
\usepackage[shortlabels]{enumitem}
\usepackage{tikz}
\usetikzlibrary{calc}

\usepackage{braket}

\mathchardef\mhyphen="2D
\newcommand{\Ex}{\mathop{\mathbb{E}}}

\newcommand{\Prob}[2]{{\Pr_{#1}\left[ {#2} \right]}}

\newcommand{\cProb}[3]{{\Pr_{#1}\left[ \left. #3 \;\right\vert #2 \right]}}

\renewcommand{\Ex}{\mathop{\mathbb{E}}}
\newcommand{\Expect}[2]{{\Ex_{#1}\left[ {#2} \right]}}

\newcommand{\cExpect}[3]{{\Ex_{#1}\left[ \left. #3 \;\right\vert #2 \right]}}

\newcommand{\eps}{\varepsilon}
\newcommand{\inr}{\in_R}

\newcommand{\defeq}{\stackrel{def}{=}}
\newcommand{\defeqq}{\coloneqq}

\newcommand{\qlrq}{\quad \Longleftrightarrow \quad}

\mathchardef\mhyphen="2D

\newcommand{\wh}[1]{\widehat{#1}}
\newcommand{\wt}[1]{\widetilde{#1}}
\newcommand{\ol}[1]{\overline{#1}}

\newcommand{\cI}{\mathcal{I}}

\renewcommand{\cP}{\mathcal{P}}

\renewcommand{\C}{\mathbb{C}}

\newcommand{\N}{\mathbb{N}}
\renewcommand{\R}{\mathbb{R}}

\crefname{claim}{claim}{claims}
\Crefname{claim}{Claim}{Claims}
\crefname{algorithm}{algorithm}{algorithms}
\Crefname{algorithm}{Algorithm}{Algorithms}
\crefname{algocf}{algorithm}{algorithms}

\makeatletter
\if@cref@capitalise
\crefname{claim}{Claim}{Claims}
\else
\crefname{claim}{claim}{claims}
\fi
\makeatother

\newcommand{\norm}[1]{\left\| #1 \right\|}

\renewcommand{\set}[1]{{\left\{ #1 \right\}}}

\newcommand{\sett}[2]{\left\{ #1 \;\middle\vert\; #2 \right\}}

\newcommand{\abs}[1]{{\left\lvert{#1}\right\rvert}}

\newcommand{\supp}{\mathrm{Supp}}

\newcommand{\power}[1]{\{0,1\}^{#1}}
\newcommand{\pown}{\power{n}}
\newcommand{\powern}{\pown}

\DeclarePairedDelimiter{\parentheses}{\lparen}{\rparen}
\newcommand{\ps}[1]{\parentheses*{#1}}
\newcommand{\pss}[1]{(#1)}

\newcommand{\pfrac}[2]{\parentheses*{\frac{#1}{#2}}}

\newcommand{\ceil}[1]{\left\lceil #1 \right\rceil}
\newcommand{\floor}[1]{\left\lfloor #1 \right\rfloor}

\newcommand{\ketbra}[2]{\ket{#1}\bra{#2}}
\newcommand{\ketbraa}[1]{\ketbra{#1}{#1}}

\renewcommand{\braket}[2]{\left\langle{#1} \middle|{#2}\right\rangle}

\newclass{\TISP}{TISP}
\newclass{\RTISP}{RTISP}

\newclass{\RTISPs}{R^*TISP}
\newclass{\BPTISP}{BPTISP}
\newclass{\BPTISPs}{BP^*TISP}
\newclass{\prBPTISPs}{prBP^*TISP}

\newclass{\BPSs}{BP^*SPACE}
\newclass{\RSs}{R^*SPACE}

\newclass{\prBPSs}{prBP^*SPACE}
\newclass{\prRSs}{prR^*SPACE}

\newclass{\searchBPP}{SearchBPP}
\newclass{\searchP}{SearchP}

\renewclass{\promiseBPP}{prBPP}
\renewclass{\promiseRP}{prRP}
\newclass{\prRTISPs}{pr\RTISPs}

\newclass{\prZPL}{prZPL}
\newclass{\prBPL}{prBPL}
\newclass{\prRL}{prRL}
\newclass{\prL}{prL}

\newclass{\searchNL}{SearchNL}
\newclass{\searchBPL}{SearchBPL}
\newclass{\searchBPLs}{SearchBP^*L}
\newclass{\searchZPLs}{SearchZP^*L}
\newclass{\searchRLs}{SearchR^*L}
\newclass{\searchRL}{SearchRL}
\newclass{\searchL}{SearchL}

\newclass{\RLs}{R^{*}L}
\newclass{\BPLs}{BP^{*}L}
\newclass{\ZPLs}{ZP^{*}L}

\newclass{\BPPLs}{BP^*PL}
\newclass{\RPLs}{R^*PL}
\newclass{\prBPPLs}{prBP^*PL}
\newclass{\prRPLs}{prR^*PL}

\newclass{\prBPLs}{pr\BPLs}
\newclass{\prRLs}{pr\RLs}

\newclass{\searchBPSACE}{Search\BPSPACE}
\newclass{\searchDSPACE}{Search\DSPACE}

\newclass{\almost}{almost\mhyphen}
\newclass{\almostLw}{almost\textsubscript{w}\mhyphen L}
\newclass{\almostLs}{almost\textsubscript{s}\mhyphen L}
\newclass{\almostL}{\almost L}

\newclass{\recursiveset}{R}
\newclass{\MAL}{MAL}

\renewcommand{\P}{\mathbf{P}}

\usepackage{physics}
\let\abs\saveabs

\usepackage{amsmath,amssymb}
\usepackage{tikz}
\usetikzlibrary{positioning,fit,backgrounds,calc,arrows.meta,decorations.pathmorphing}

\usepackage{yhmath}
\usepackage{bbm}

\usepackage{float}
\usepackage[ruled,vlined]{algorithm2e}
\SetKwFor{For}{for}{do}{}
\SetKwIF{If}{ElseIf}{Else}{if}{then}{else if}{else}{}
\SetKwInput{KwWitness}{Witness}
\SetKw{Accept}{accept}
\SetKw{Reject}{reject}
\SetKw{Continue}{continue}
\SetKw{Print}{print}

\usepackage[textsize=tiny]{todonotes}

\usepackage[normalem]{ulem}

\newcommand{\Cones}{\mathsf{CONES}}
\newcommand{\ProjUSLHinstance}{\ensuremath{\mathsf{ProjUSLH}_{q}(k,d)}}
\newcommand{\ProjUSLHinstancee}{\ensuremath{\mathsf{ProjUSLH}}}

\newclass{\stoqMA}{StoqMA}
\newlang{\LH}{LH}

\newlang{\ProjUSLH}{ProjUSLH}

\newlang{\YES}{Yes}
\newlang{\NO}{No}
\crefname{AlgoLine}{line}{lines}
\Crefname{AlgoLine}{Line}{Lines}

\begin{document}
\title{Frustration Free Stoquastic Local Hamiltonian \\ with Sub-Constant Gap is in $\NP$}
\author{
    Eshan Chattopadhyay\thanks{Cornell University. Email: \texttt{eshan@cs.cornell.edu}.}
    \and
    Oren Renard\thanks{Cornell University. Email: \texttt{oren@cs.cornell.edu}.}
    \and
    Nicholas Spooner\thanks{Cornell University. Email: \texttt{nspooner@cornell.edu}.}
}
\maketitle

\begin{abstract}
We continue the study of the Stoquastic Local Hamiltonian problem, a physically motivated restriction of the $\QMA$-complete Local Hamiltonian problem \cite{KSV02}. For the $\beta$-gapped, frustration-free case, \citet{BBT06} showed that the problem is $\MA$-complete when $\beta= \tfrac{1}{\poly\ps{n}}$. \citet{AG19} derandomized this algorithm and proved membership in $\NP$ for constant gap $\beta = \Omega\ps{1}$.

We present an improved algorithm and analysis, establishing membership in $\NP$ even when $\beta = \Omega \ps{\frac{1}{{\log\log n}}}$.
We complement our result with an explicit example demonstrating why the analysis does not extend directly to $\beta = o(1/\log\log n)$.
\end{abstract}

\tableofcontents
\section{Introduction}
The $k$-Local Hamiltonian problem is the canonical complete problem for $\QMA$ \cite{KSV02}, the quantum analogue of $\NP$, and is itself the quantum analogue of 3SAT.
An instance is a tuple $(H, \alpha, \beta)$ where $H = \frac{1}{m}\sum_{i} H_i$ is a normalized sum of Hermitian operators, each $H_i$ acting nontrivially on at most $k$ qubits. The goal is to decide whether the ground energy of the system, namely its minimum eigenvalue $\lambda_{\min}(H)$, is at most $\alpha$ or at least $\beta$, promised that one of these holds. The quantity $\beta - \alpha$ is called the \emph{promise gap}, and the problem is $\QMA$-complete for any inverse-polynomial gap $\beta - \alpha \ge 1/\poly(n)$ \cite{KSV02}. The locality parameter $k$ plays an important role, and in this work, we focus on arbitrary constant $k$, though our results extend naturally to any locality.

% The quantum PCP conjecture \cite{aharonov2013guest} asks whether the Local Hamiltonian problem (the quantum analogue of constraint satisfaction) remains hard at constant gap; while its multiplayer-games variant has seen dramatic progress (notably $\MIP^* = \RE$ \cite{JNVWY21}), the Hamiltonian version studied here remains wide open.

Due to the intractability of the Local Hamiltonian problem in full generality, a vast line of research has studied restricted variants, seeking better upper bounds (e.g., geometrically local instances \cite{oliveira2005complexity,AGIK09-quantum-on-line}, commuting Hamiltonians \cite{bravyi2003commutative,hastings2012trivial}, translationally invariant systems \cite{bausch2017complexity,cubitt2016complexity}, and much more; see \cite{gharibian2014quantum} for a survey).

In this work, we consider two standard restrictions. 
The first is \emph{frustration-free} instances, where the ground state (the minimal eigenvector) simultaneously satisfies every local term $H_i$ (equivalently, $\alpha = 0$). This is a well-studied condition known to preserve hardness in many settings \cite{BBT06,cubitt2016complexity,AGIK09-quantum-on-line}.

The second, and more central to this work, is the restriction to \emph{Stoquastic} Hamiltonians \cite{BBT06,BT10}, a natural ``sign-free'' variant of the Local Hamiltonian problem. 
Since each Hermitian matrix $H_i$ can have complex entries, it permits destructive interference between amplitudes, a key source of quantum hardness and the main obstacle to classical simulation. 
Stoquastic Hamiltonians eliminate this by restricting each local term $H_i$ to be real-valued, with non-positive off-diagonal entries. In this case, one can show there exists a \emph{ground state} with real, non-negative amplitudes \cite{BBT06}, rendering it tantamount to a \emph{classical} probability distribution amenable to classical Markov chain methods \cite{BBT06, AG19,bravyi2017complexity,piddock2015complexity,bravyi2006complexity,klassen2019two,aharonov2025stoqma,REG26-geometric-completeness-stoqlh}.

%
% \footnote{The formal definition is different, but reduces to the written statement. Formally, each $H_i$ is required to be real-valued with non-positive off-diagonal entries. However, shifting attention to the projectors $P_i$ onto the ground spaces of $H_i$, one can show that each $P_i$ is entrywise real and non-negative \cite{BBT06}, and that this formulation loses no generality. It is therefore the Hamiltonian $\sum_i P_i$ in this projector formulation that we work with throughout.}
% 
% This confines the ground state to have real, non-negative amplitudes.

\cite{BBT06} define the complexity class $\stoqMA(\alpha,\beta)$, a variant of $\QMA$ in which the verifier is fully classical except for a final Hadamard-basis measurement, and show that the Stoquastic Local Hamiltonian problem is complete for this class. It follows from the definition of the class that
\[
\MA \subseteq \stoqMA(\alpha, \beta) \subseteq \QMA.
\]
% The class $\stoqMA$ also admits a natural computational interpretation as an intermediate model between $\MA$ and $\QMA$.\footnote{Concretely, \cite{BBT06} showed that it is the class of problems decidable by an ``$\MA$-style'' circuit---with a nondeterministic witness, random coins, and only Toffoli gates---but with the final measurement performed in the $\{\ket{+},\ket{-}\}$ basis, which clarifies its position between $\MA$ and $\QMA$.}
% 
Intuitively, it seems that $\stoqMA$ lies closer to $\MA$, since the $\stoqMA$ verifier is ``barely quantum''. Indeed it is conjectured that $\stoqMA = \MA$ \cite{BBT06,aharonov2025stoqma}, though this remains open.\footnote{It is known that $\stoqMA \subseteq \mathbf{AM}$ \cite{bravyi2006complexity}.}

% Following \cite{BBT06,AG19}, we denote by $\ProjUSLH_q(k, d, \alpha, \beta)$ the corresponding Projective Uniform Stoquastic local Hamiltonian gap problem over $q$-dimensional qudits, which is also complete for $\stoqMA(\alpha, \beta)$.\footnote{The Stoquastic LH problem was introduced by \cite{BBT06}. For frustration-free instances, \citet{AG19} further introduced the \emph{uniform} variant, in which each local projector term decomposes into orthogonal \emph{subset states} (see \cref{sec:preliminaries}). This restriction simplifies the analysis while preserving $\MA$-completeness.} The parameter $d$ bounds the number of terms involving any given qudit; for clarity, the reader can assume it to be a constant $d=O\ps{1}$ throughout the paper.\footnote{The assumption $d=O(1)$ is justified by Kitaev's circuit-to-Hamiltonian reduction \cite{KSV02}; see also \cite{AG19}.}

\citet{BBT06} showed that the \emph{frustration-free projective} Stoquastic Local Hamiltonian problem (in which every local term is a projector and $\alpha = 0$) is $\MA$-complete at inverse-polynomial soundness $\beta = 1/\poly(n)$. The result was later extended to the \emph{uniform projective} case \cite{AG19}.\footnote{A projector is called uniform if it can be decomposed into sum of orthogonal uniform subset state projectors. See \cref{sec:preliminaries} for a formal definition.}
Most notably, these yielded the first $\MA$-complete problem(s), offering a surprising quantum lens through which to study $\MA$.
We focus on this setting and denote the problem over $q$-dimensional qudits by $\ProjUSLH_q(k,d,\alpha,\beta)$, where $\alpha$ and $\beta$ are the YES- and NO-case energy thresholds.
\begin{theorem}[{\cite{BBT06,AG19}}]\label{thm:bbt main}
    There exist constant integers $q,k,d$ such that $\ProjUSLH_q(k, d, 0, \tfrac{1}{\poly(n)})$ is $\MA$-complete.
\end{theorem}

Remarkably, \citet{AG19} showed that at constant soundness $\beta >0$, the problem can be decided in non-deterministic polynomial time, namely $\ProjUSLH_q(k, d, 0, \beta = \Omega(1)) \in \NP$.
Remarkably, this implies that a stoquastic PCP theorem, i.e., gap amplification for $\ProjUSLH(k,0, \cdot)$, would collapse $\MA = \NP$, a major open problem in derandomization \cite{NW94,IW97,IKW02,KVM02}. 

\subsection{Our results}

We present an improved algorithm for $\ProjUSLH$ that pushes the tractability threshold to a sub-constant soundness while remaining within $\NP$, yielding strict improvement over \cite{AG19}.
Formally, although our results extend naturally to arbitrary parameter choices, we focus on constant parameters $k, q, d = O\ps{1}$.

Our main contribution is an alternative analysis framework that yields an improved $\NP$ algorithm for subconstant soundness $\beta=o(1)$.
Our framework uses the probabilistic method (though the algorithm itself is deterministic), which yields a stronger, simpler analysis.

\begin{theorem}\label{thm:main NP general X}
    For any $k,d,q \in \mathbb{N}$ there exists $\eta$ such that $\ProjUSLH_{q}(k, d, 0, \eta\cdot  \tfrac{1}{\log \log n}) \in \NP$.
\end{theorem}

Combining the above with \cref{thm:bbt main}, we conclude that a gap amplification reduction for $\ProjUSLH_q$ from inverse-poly up to $\beta=\Omega\ps{\frac{1}{{\log\log n}}}$ would suffice to derandomize Merlin-Arthur games.

\begin{corollary}
    There exists a constant $d \in \N$ such that, if there exist $k',d',q' \in \N$ such that there is a reduction from $\ProjUSLH_{q=2}(k=6, d, 0, \tfrac{1}{\poly(n)})$ to $\ProjUSLH_{q'}(k', d', 0, \eta' \cdot \tfrac{1}{\log \log n})$ for large enough constant $\eta' = \eta'(k',d',q')$, then $\MA = \NP$.
\end{corollary}

Finally, we give evidence that the soundness bound in \cref{thm:main NP general X} is tight for our current algorithmic approach. At a high level, the verifier searches for a short sequence of local terms that leads from a proposed witness string to a contradiction. Its running time depends not only on the length of such sequences, but also on the number of candidates that appear promising according to the progress measure used in our analysis. In \cref{sec:example delta valid cones}, we construct explicit frustrated instances admitting exponentially many such candidates. Consequently, this progress measure alone cannot prune the search sufficiently to extend \cref{thm:main NP general X} to $\beta=o\ps{1/\log\log n}$; such an extension would require a new approach.

\subsection{Technical Overview}\label{sec:techincal overview}
A critical component in the analysis of stoquastic Hamiltonians is the graph-based framework of \cite{BBT06}. We use its simplified undirected form that was presented by \cite{AG19}.\footnote{The undirected and directed variants correspond to the uniform and non-uniform variants of Stoquastic Local Hamiltonian, respectively. Both variants are $\stoqMA$-complete and therefore fully characterize the problem.}

For this overview, an instance $\cP = \set{P_i}_{i \in [m]}$ of $m$ many $k$-local Projector Uniform Stoquastic terms over $n$ qubits consists of projectors $P \in \cP$ that are (a) entrywise real and non-negative, and (b) $k$-local. 
To simplify the exposition, we ignore the ``uniform'' restriction here, though it remains essential for the proof (see the full definition at \cref{def:projUSLH}).

We embed $\cP$ into an exponentially large graph $G = (\powern, E)$, with vertices corresponding to computational basis elements $x\in \powern$ and edges defined by
\[
(x,y) \in E \qlrq \bra x \frac{1}{m} \sum_{i\in[m]} P_i \ket y > 0.
\]
\citet{BBT06} introduced the notion of BAD strings $x \in \powern$, that are defined as those satisfying $\bra x P_i \ket x = 0$ for some term $P_i$.
Subsequently, they showed that these strings sharply characterize the perfect-completeness gap problem, where YES instances are frustration-free while NO instances are $\eps$-frustrated on average:
\begin{enumerate}
    \item[(YES)] if $\ket\theta$ is a ground state that is supported on the set $\supp(\ket\theta) = \sett{x}{\braket{\theta}{x} \neq 0}$, then the connected component of $\supp(\ket\theta)$ in $G$ consists entirely of GOOD strings;
    
    \item[(NO)] 
    every connected component of $G$ contains a BAD string, and a random walk from any vertex hits a BAD string with high probability.
\end{enumerate}
This yields a natural $\MA$ verifier: it expects a witness $x \in \supp(\ket\theta)$, walks randomly on $G$, and rejects upon encountering a BAD string.

\citet{AG19} subsequently derandomized this verifier. Their key insight is that, in NO instances which are $\eps$-frustrated for constant $\eps >0$, one can deterministically locate a BAD string because \emph{frustration implies expansion}. 
To this end, they relied on two main lemmas in the analysis of such NO instances:

First, they showed that every state $\ket\varphi$ is ``$\eps/2$-frustrated'' by at least an $\eps/2$ fraction of the terms. This follows from a simple averaging argument: there exist many terms $P_i \in \cP$ that satisfy
\begin{align}\label{eq:frustrationa}
    \norm{P_i \ket{\varphi}}_2^2 \leq 1-\eps/2.
\end{align}
Second, a short algebraic calculation shows that any such $P_i$ \emph{expands} the support of $\ket{\varphi}$, when $\ket{\varphi}$ is a ``uniform subset state''.\footnote{We say $\ket{\varphi}$ is a subset state if $\ket{\varphi} = \tfrac{1}{\sqrt{|S|}} \sum_{x \in S} \ket{x}$ for some $S \subseteq \powern$.}
At a high level, because the terms are stoquastic projectors, i.e., entrywise non-negative, frustration implies growth of the support size of entry-wise non-negative states. Since \cref{eq:frustrationa} gives step frustration at least $\eps/2$, the frustration--expansion lemma yields
\begin{align}\label{eq:expansion}
    \abs{\supp(P_i \ket{\varphi})} \ge (1 + \eps/4)\abs{\supp(\ket{\varphi})}.
\end{align}
Then \cite{AG19} organizes the expanding projectors into $\ell=O(1/\eps^2)$ layers, each of which contains a sequence of terms that act on non-overlapping qubits. Then, starting from the state $\ket{\varphi_0} = \ket{x}$ and applying the terms sequentially, as long as no BAD string is encountered, each term (or layer) expands the support; since the support has size at most $2^n$, a BAD string must eventually be reached (which always exists for a frustrated instance).

Tracing backward from a term witnessing this BAD string, i.e. $P$, gives a lightcone (denoted $C_P$), whose size is bounded by $\abs{C_P}  \approx k^{\ell}$. Thus the BAD string lies within radius $\abs{C_P}$ of the witness $x$ in the graph $G$, and their verifier finds it by exploring the entire surrounding ball in the graph $G$.
Their final argument is therefore to bound the number of vertices in this ball:
\begin{align}\label{eq:ball volume}
    \abs{\mathrm{Ball}_G(x,\, \abs{C_P})} \le n^{\abs{C_P}} \approx n^{k^\ell}.
\end{align}
The issue is that plugging $\ell \approx 1/\eps^2$ and enforcing polynomial runtime requires \cref{eq:ball volume} to be polynomial, forcing $\eps = \Omega(1)$ to be constant.

\paragraph{Improvement to $\eps \ge \tfrac{1}{\sqrt{\log \log n}}$.}
As a first step, we observe that enumerating the entire ball around $x$ is wasteful. Instead of enumerating the entire ball, we can restrict the walk to terms that form a \emph{lightcone}, of which there are far fewer:
\begin{align}\label{eq:cones trivial bound}
    \abs{\Cones_{P, \ell}} \le 2^{{kd}^{\ell}}.
\end{align}
This bound uses two separate properties of the layered lightcones from \cite{AG19}: terms within each layer act on disjoint sets of qubits, while the lightcone construction retains from each layer only terms that overlap a retained term in the next layer closer to the root.
Since all terms are $k$-local, the number of candidates overlapping with any fixed term is bounded by $kd$, and iterating this count across $\ell$ layers yields the bound above.

Enumerating $\abs{\Cones_{P,\ell}}$ in polynomial time would enforce \cref{eq:cones trivial bound} to be at most $\poly(n)$, and plugging $\ell=O(1/\eps^2)$ into \cref{eq:cones trivial bound} yields $\eps \ge \tfrac{1}{\pss{\log\log n}^{1/2}}$ and therefore an immediate improvement over the earlier $\NP$ algorithm for soundness $\eps = \Omega(1)$.

\paragraph{Improvement to $\eps \ge \tfrac{1}{\log \log n}$.} 
We next present how to improve the algorithm to support quadratically smaller $\eps$.
% \red{The square-root loss comes from the layered analysis of \cite{AG19}, which requires $\ell \approx 1/\eps^2$ layers. Our verifier is likewise deterministic; only the analysis is randomized, replacing these layers by a sufficiently long random sequence of terms.} We then show that, although the sequence itself may be long, only a short overlapping subsequence is relevant for reaching a BAD string.
%
The reason for the square root loss comes from the number of layers in the argument above: the deterministic expansion argument of \cite{AG19} requires $\ell \approx 1/\eps^2$ layers. We avoid this loss by replacing the carefully chosen (but otherwise ``worst-case'') layers with a sufficiently long random sequence of terms.
(Note that this step only appears in the analysis; the $\NP$ algorithm remains deterministic.) We then show that---similarly to \cite{AG19}---although the sequence itself may be long, only a short overlapping subsequence is relevant for reaching a BAD string.

In the NO case, \emph{every} state $\ket{\varphi}$ has frustration at least $\eps$. This means that, writing $\gamma_i = 1 - \norm{P_i \ket{\varphi}}_2^2$, we have
\[
\Expect{i\inr [m]}{\gamma_i}
= \frac{1}{m}\sum_{i=1}^m \gamma_i
= 1 - \bra{\varphi}\frac{1}{m}\sum_{i=1}^m P_i\ket{\varphi}
\ge \eps.
\]
Thus, choosing $P_i$ uniformly at random gives frustration at least $\eps$ in expectation. For a subset state whose support contains no BAD strings, the frustration/expansion lemma (\cref{eq:expansion}) therefore guarantees expansion in expectation as well.
This gives rise to a very natural approach: choose independently and uniformly at random $T$ terms, for some integer $T$ to be determined:
\[
P_{i_1},\ldots,P_{i_T} \inr \cP.
\]
Fix an arbitrary witness $x\in\powern$. Starting from $S_0=\{x\}$, we apply the terms sequentially. We denote the $t$-th step support and frustration as:
\begin{align*}
    S_t &\defeqq \supp \ps{ P_{i_t} \ket{S_{t-1}}}, \\
    \gamma_t &\coloneqq 1- \norm{ P_{i_t}\ket {S_{t-1}} }_2^2.
\end{align*}
Since the instance is $\eps$-frustrated and $P_{i_t}$ is chosen uniformly, we have that by definition, the expected frustration at any step is $\Expect{P_{i_t}}{\gamma_t} \geq\eps.$
Moreover, as long as we have not encountered a BAD string, the frustration/expansion lemma (\cref{eq:expansion}) implies that the sets $S_1,\ldots, S_T$ are gradually expanding:
\[
    \abs{S_t}
    \geq
    \left(1+\frac{\gamma_t}{2}\right)|S_{t-1}|.
\]
Thus, if no BAD string were reached, one can show that the expected support after $T$ steps would satisfy
\begin{align}\label{eq:final expected frustration}
\Expect{}{ \abs{ S_T }}
    \geq
    \exp \ps{ \frac{1}{4}
        \sum_{t=1}^T \Expect{}{ \gamma_t} }
    \geq
    \exp\ps{ \frac{\eps T}{4} }.
\end{align}
Therefore, choosing $T = 16 n/\eps$, this lower bound is at least $2^{2n}$, which exceeds the $2^n$ strings in the entire space. Hence some sequence of terms must reach a BAD string.

However, enumerating all sequences of length $T$ would be far too expensive. The key point is that most terms in such a sequence are irrelevant for producing the final BAD string. To see this, for every step $t$, consider the (reversed order) sequence
\[
    P_{i_t},P_{i_{t-1}},\ldots,P_{i_1}.
\]
Now beginning from the first term $P_{i_t}$, iteratively retain only terms that overlap some term already retained. Denote the resulting sequence as
\[
    P_{i_t'}'=  P_{i_t} , P_{i_{t-1}'}' , \ldots, P'_{i_1'},
\]
and denote the length of this subsequence as $D_t$.

This is precisely the backward lightcone of the root term within the random sequence. In contrast to the structured lightcone that emerges from the analysis of \cite{AG19}, our lightcone's only structure requirement is that each term overlaps a previously chosen term (but not necessarily in layers).

It is not too hard to show that the expected length of the retained overlapping sequence is at most exponentially large in $1/\eps$: 
\[
    \Expect{}{ D_t} \approx 2^{O(1/\eps)}.
\]
Therefore, setting $\ell_\eps = \tfrac{2}{\eps} \Expect{}{D_t}$, it follows by Markov's inequality that with high probability the retained overlapping sequence has length at most $\ell_\eps$:
\[
    \Prob{}{ D_t>\ell_\eps } \leq\frac{\eps}{2}.
\]
We then show that even after \emph{chopping} the random sequence whenever $D_t>\ell_\eps$, namely removing some terms $P_{i_t}$, the expected frustration is still relatively high, i.e. $\ge \eps/2$. The preceding expansion argument (\cref{eq:final expected frustration}) therefore continues to work, with only a constant-factor loss. In particular, the chopped sequence must still reach a BAD string. On the other hand, every term that survives the chopping has a bounded (backward) lightcone of size at most 
\[
\ell_\eps \approx \exp(1/\eps),
\]
in contrast to our previous argument bounding it by $\approx \exp(1/\eps^2)$. 
This allows us to choose $\eps > 1/\log \log n$ while still enumerating all overlapping sequences of length $\ell_\eps$ in polynomial time.

\paragraph{Tightness of the algorithm.}
One may hope to improve the algorithm by \emph{pruning} overlapping sequences that do not make progress. In particular, we could enumerate only sequences in which every applied term has noticeable frustration. We show that this does not help:
our last contribution is demonstrating an explicit frustrated instance, which has exponentially many overlapping sequences, that all appear equally promising according to this criterion (namely, all of them are sequentially frustrated).

Let us describe the example at a high level. The instance is arranged on a complete $k$-ary tree over the nodes $V = [n]$, where the qudits live over the alphabet $\Sigma = \set{0, \dots, q-1}$ for $q \geq 3$. We associate a single term $P^{(v)}$ with each node $v\in V$, that acts as follows:
\begin{itemize}
    \item it acts only on the qudit $v$ and its children $v[1],\ldots, v[k]$.
    \item $P^{(v)}$ projects the $v$-th qudit to the subset state $|A\rangle$, for some $A \subseteq \Sigma$.
    
    \item $P^{(v)}$ projects each child register into the fully uniform state $\ket{\Sigma}$.
\end{itemize}
By choosing $|A|\approx(1-\delta)q$, one can show that these two requirements are slightly incompatible, making the instance $\Omega(\delta/k)$-frustrated, and therefore a NO instance. 

Now start from the witness $x=0^n$. For any rooted subtree, list its terms top down from the distinguished root, as in the rooted overlapping-sequence convention. The corresponding operator product applies them in reverse, bottom-up order. Immediately before applying $P^{(v)}$, register $v$ is therefore still in state $|0\rangle$, and hence
\[
\bigl\|P^{(v)}|\psi\rangle\bigr\|_2^2 
\leq1-\delta.
\]
Thus, every term frustrates the state when it is applied, and every such rooted subtree gives a sequentially $\delta$-frustrated overlapping sequence.

The problem is that there are many possible subtrees. At every reached node, we may independently choose which of its $k$ children to continue exploring. Among sequences containing at most $L$ applied terms, the construction gives at least $2^{\Omega(L)}$ different candidates. Taking $L=\ell_\eps \approx \exp(1/\eps)$ therefore yields
\[
2^{\Omega(L)} = \exp\ps{\exp\ps{\ps{1/\eps}}}.
\]
Although one of them is guaranteed to reach a BAD string, local frustration gives no way to distinguish it from the others. Hence a pruning strategy based on local frustration cannot improve the $\eps = \Theta(1/\log\log n)$ threshold of our algorithm.

\subsection{Discussions}
% \paragraph{Derandomizing the Analysis.}
% Our random selection of indices in \cref{sec:chop} is helpful to gurantuee that on average, the frustrated at each step $\gamma_t$ is at least $\ge \eps$, which is eventually necessary to gurantuee gradual expansion of the set $S_T$. If it an interesting research direction to try select those indices $i_1, \ldots i_T \sim \cI$ from some different distrubtion that have smaller support size compared to the unifmr $m^T$. Such a distribution might assist ...

\paragraph{Geometric constraints.}
Our techniques combine naturally with the geometric constraints of a $D$-dimensional lattice, yielding an algorithm for exponentially smaller soundness parameters, up to $\beta = \tfrac{1}{(\log n)^{O\ps{1/D}}}$. 

However, the ground energy of any (general) Local Hamiltonian problem on a $D$-dimensional lattice can be approximated deterministically in polynomial time (namely, in $\P$) to additive accuracy $\tfrac{1}{\poly (\log n)}$. Indeed, as noted by Hastings and Terhal (see \cite{AGIK09-quantum-on-line}), one can decompose the Hamiltonian into blocks over $O(\log n)$ qudits, find the ground-state energy within each block in exponential time $2^{O(\log n)} = \poly (n)$, and then deduce an $\approx (\log n)^{-1/D}$ additive approximation for the original instance. See \cite{aharonov2013guest} for additional details.

It is a curious coincidence that the deterministic block-decomposition technique and the frustration/expansion framework of \cite{AG19} both reach $\tfrac{1}{\poly(\log n)}$ approximation scales.

\paragraph{On the importance of uniformity-preserving reductions.}
Our results build on \citet{AG19} and are therefore restricted to stoquastic Local Hamiltonian instances satisfying the \emph{uniformity} condition. As observed there, extending the frustration–expansion framework to non-uniform weights is challenging because frustration may depend on global properties of the Hamiltonian. Our improvements thus further motivate understanding the distinction between the uniform and non-uniform variants.

\section{Preliminaries}\label{sec:preliminaries}
\subsection{Definitions}\label{sec:definitions}
We follow definitions and notations from \cite{AG19}. 

Let $\Sigma = [q]$ be an alphabet for $q\in \N$.
For any subset $S \subseteq \Sigma^n$, the Uniform Subset State $\ket S$ is defined as
\[
\ket S \defeq \frac{1}{\sqrt{\abs S}} \sum_{x\in S} \ket x .
\]
The support of a quantum state $\ket \varphi$ is defined as
\[
\supp(\ket \varphi) \defeq \sett{x}{\braket{x}{\varphi} \neq 0}.
\]
The Uniform Support Subset State of a quantum state $\ket \varphi$, denoted as $\ket{\wh \varphi}$, is defined as the uniform subset state over the support of the state, namely
\[
\ket{\wh \varphi} \defeq \ket{\supp(\ket \varphi)}.
\]
\begin{definition}\label{def:k-local uni stoc proj}
    A linear operator $P \in \C^{q^n \times q^n}$ is \emph{$k$-local Uniform Stoquastic Projector} if it satisfies the following:
    \begin{enumerate}
        \item Stoquastic Projector: $P$ is entrywise real and non-negative, namely $P \in (\R^{\ge 0})^{q^n \times q^n}$.
        
        \item $k$-local: $P$ acts non-trivially on at most $k$ qubits, namely $P = P'_{S} \otimes I_{[n] \setminus S}$ for some $\abs S \le k$.

        \item Uniform: the non-trivial part can be written as $P' = \sum_{j\in[r]} \ketbra{T_j}$, where each $\ket{T_j}$ is a uniform subset state, and the subset states $\ket{T_j}$ are orthonormal $\braket{T_i}{T_j} = \delta_{ij}$.

        % \item Uniform: $P = \sum_{j\in[q]} \ketbraa{T_j}$, where each $\ket{T_j}$ is a uniform subset state, and $\braket{T_j}{T_k} = 0$ for $j \neq k$.
    \end{enumerate}
\end{definition}
For example, if $P$ is a Uniform Stoquastic Projector acting non-trivially on exactly the first $k$ qubits, there exist an integer $r \le q^k$ and sets $T_{j} \subseteq [q]^{k}$ such that
\[
P = \sum_{j\in [r]} \sum_{z\in [q]^{n-k}} \ps{ \ketbraa{T_{j}} }_{[k]} \otimes ( \ketbraa{z} )_{[k+1,n]}.
\] 
The support of a $k$-local linear operator $P \in \C^{q^n\times q^n}$, denoted by $\supp(P)$, is the set of qudits $S \subseteq [n]$ on which it acts non-trivially.

\begin{definition}\label{def:projUSLH instance}
    We define $\ProjUSLHinstance$ as the set of all collections $\mathcal{P} = \set{P_i}_{i \in [m]}$ such that each $P_i$ is a $k$-local Uniform Stoquastic Projector over $n$ qudits with alphabet $\Sigma = [q]$, and each qudit participates in at most $d = O(1)$ terms.
\end{definition}

\begin{definition}\label{def:projUSLH}
The \emph{Projection Uniform Stoquastic Local Hamiltonian} promise problem is defined as
\[
\ProjUSLH_q(k, d, \alpha, \beta) \defeq \YES_{q, k, d, \alpha} \cup \NO_{q,k, d, \beta}.
\]
We omit the subscripts and abbreviate $\YES, \NO$ when the parameters are clear from context.

An instance $\cP \in \ProjUSLHinstance$ belongs to one of the sets if:
\begin{enumerate}
    \item $\cP \in \YES$: there exists a ``ground state'' $\ket{\psi}$ such that $\bra{\psi} \tfrac{1}{m} \sum_{i \in [m]} P_i \ket{\psi} \geq 1 - \alpha$;
    
    \item $\cP \in \NO$: for all states $\ket{\psi}$ it holds that $\bra{\psi} \tfrac{1}{m} \sum_{i \in [m]} P_i \ket{\psi} \leq 1 - \beta$.
\end{enumerate}
Any instance not belonging to either $\YES$ or $\NO$ is outside the promise.
\end{definition}

Our results work with an arbitrary alphabet size $q$, but for simplicity one should think of a constant-size alphabet, e.g., binary $q=2$. We omit $q$ from the notation when it is clear from context.

\subsection{The Associated Graph and Its Analysis}
Following \citet{BBT06,AG19}, we associate each instance of $\ProjUSLH$ with a graph. Given $\mathcal{P} = \set{P_i}_{i \in [m]} \in \ProjUSLHinstance$, define the undirected graph $G = (\Sigma^n, E)$ as follows:
\[
(x,y) \in E
\qlrq
\exists i \in [m] \; \text{ such that } \; \bra x P_i \ket y > 0.
\]

\begin{definition}[\cite{BBT06}]
    A string $x \in \Sigma^n$ is called $P_i$-\emph{BAD} if $\bra xP_i \ket x = 0$.
    Similarly, $x$ is \emph{BAD} if it is $P_i$-BAD for some $P_i \in \mathcal P$. A string $x \in \Sigma^n$ is \emph{GOOD} if it is not BAD.
\end{definition}
From here on we focus on analyzing instances $\cP$ of the gap problem $\ProjUSLH_q(k, d,0, \eps)$, namely those that admit either perfect completeness or non-trivial soundness $\eps > 0$. We introduce the high-level analysis of such graphs.

% \subsubsection{Yes case}
% In this case, we rely on earlier work:

In the YES case, we rely on earlier work:

\begin{lemma}[{\cite[Appendix A.2.1]{BBT06}}]\label{thm:BBT-yes-case connected componenet only GOOD}
    Consider a $\YES$ instance $\cP$ of $\ProjUSLH_q(k,d, 0, \eps)$, for arbitrary $\eps > 0$.
    
    Let $\ket \theta$ be the ground state of $\cP$. Then the connected component of every $w \in \supp(\ket \theta)$ in the graph of $\cP$ contains no BAD strings.
\end{lemma}

% \subsubsection{No case}
% In this case, \cite{AG19} demonstrated the frustration vs. expansion lemma, which plays a critical cole in finding BAD strings.

In the NO case, \cite{AG19} demonstrated the frustration vs. expansion lemma, which plays a critical role in finding BAD strings.

\begin{lemma}[\cite{AG19}]\label{thm:ag19 frustration expansion}
    Let $P\in \R^{\abs \Sigma^n \times \abs \Sigma^n}$ be a uniform stoquastic projector. Let $S \subset \Sigma^n$ be an arbitrary set that contains no $P$-BAD strings.

    Denote the frustration $\gamma = 1 - \norm{P \ket S}_2^2$. Then frustration implies expansion:
    \[
    \abs{\supp (P \ket S)} \ge 
    (1 + \gamma/2) \abs S.
    \]
\end{lemma}

% %%%%%%%%%%%%%\input{soruces/lightcone-alg}
\newcommand{\overlapseq}{\mathsf{OverlapSubseq}}

\section{Our Algorithm}\label{sec:chop}
In this section, we present a probabilistic analysis as an alternative to the deterministic one of \cite{AG19}. To this end, instead of relying on the original notion of \emph{structured} lightcones from \cite{AG19}, we introduce a new notion of \emph{unstructured} lightcones, namely of overlapping sequences.

\begin{definition}
    Let $\cP \in \ProjUSLHinstance$.
    Let $(Q_1, \ldots, Q_t)$ be a sequence of terms in $\cP$.
    The sequence is called \emph{overlapping} if every term $Q_{j}$ for $j \ge 2$ overlaps some earlier term $Q_{j'}$ for $j' < j.$
\end{definition}

\begin{theorem}\label{thm:improved-np}
$\ProjUSLH_q(k,d,0, \eta \cdot \frac{1}{\log \log n}) \in \NP$ for some constant $\eta = \eta(k,d,q)$.
\end{theorem}

\begin{algorithm}[H]
\caption{}
\label{alg:verifier np loglogn}
\KwIn{An instance $\mathcal{P} = \set{P_i}_{i\in[m]}$, and $\eps > 0$.}
\KwWitness{$x \in \Sigma^n$.}
\BlankLine
\For{all root terms $P \in \mathcal P$, and sequence length $\ell \leq O(q^{32kd/ \eps})$}{
    \For{all overlapping sequences $(P, P_{i_1}, \ldots, P_{i_\ell})$}{
        \For{all reachable strings $y \in \supp(P_{i_1} \cdots P_{i_\ell} \ket x)$}{
            \If{$y$ is $P$-BAD}{
                \Return reject\;
            }
        }
    }
}
\Return accept\;
\end{algorithm}

\begin{proof}
The verifier is described in \cref{alg:verifier np loglogn}.

\paragraph{Correctness.}
Consider first a \(\YES\) instance. Let \(\ket{\theta}\) be a frustration-free ground state, and suppose that the witness \(x\) belongs to \(\supp(\ket{\theta})\). By \cref{thm:BBT-yes-case connected componenet only GOOD}, the connected component of \(x\) consists entirely of GOOD strings. Every string examined by the verifier is reachable from \(x\), and is therefore GOOD. Hence, the verifier accepts.

Now consider a \(\NO\) instance. By \cref{lem:small-bad-sequence}, for every \(x\in\Sigma^n\), there exist a term \(P\in\cP\), and an overlapping sequence $(P,P_{i_1},\ldots,P_{i_\ell})$
such that $\supp\left(P_{i_1}\cdots P_{i_\ell}\ket{x}\right)$ contains a \(P\)-BAD string \(y\), where
\[
    \ell
    \leq \ell_\eps
    \leq q^{32kd/\eps}.
\]
The verifier enumerates every such \(P\), length \(\ell\), and overlapping sequence, and therefore eventually examines this particular sequence and finds $y$. Therefore the verifier always rejects.

\paragraph{Time complexity.}
We bound the cost of each stage of the verifier.

\begin{enumerate}
    \item Fix \(\ell\leq\ell_\eps\) and \(P\in\cP\). After the root and $r-1$ applied terms have been selected, there are $r$ retained terms in total. The next term must overlap at least one of them, so there are at most $kdr$ choices. Thus, the number of overlapping sequences with $\ell$ applied terms rooted at $P$ is at most
    \begin{align*}
        &\leq \prod_{r=1}^{\ell} kdr 
        \leq (kd\ell)^\ell 
        \leq \exp \left( O\left( \ell_\eps \log( kd \ell_\eps ) \right)\right) 
        \leq \exp \left( \exp\left(O\left(\frac{kd\log q}{\eps}\right)\right)\right).
    \end{align*}

    \item Fix an overlapping sequence \((P,P_{i_1},\ldots,P_{i_\ell})\). Each applied projector acts on at most $k$ symbols and can increase the current support size by a factor of at most $q^k$. Consequently, the final support set can be enumerated in time
    \[
        \abs{ 
        \supp\ps{ P_{i_1}\cdots P_{i_\ell}\ket{x}} 
        }
        \leq q^{k\ell}.
    \]

    \item For each enumerated string \(y\), checking whether \(y\) is \(P\)-BAD requires only the \(k\) symbols on which \(P\) acts, and hence takes time \(O(k)\).
\end{enumerate}

There are \(\abs{\cP}=\poly(n)\) choices for \(P\) and at most \(\ell_\eps\) choices for \(\ell\). Therefore, the total running time is at most
\begin{align*}
    &\le \abs{\cP} \cdot \ell_\eps\ \cdot 
    \exp \left( \exp\left(O\left(\frac{kd\log q}{\eps}\right)\right)\right)
    \cdot q^{O(k\ell_\eps)} \cdot O(k)
    \\
    &\leq
    \poly(n) \cdot q^{O(kd/\eps)} \cdot 
    \exp\left(\exp\left(
        O\left(\frac{kd\log q}{\eps}\right)
    \right)\right) \cdot q^{k \frac{2}{\eps} q^{O(kd/\eps)}} \\
    &\le  \exp \ps{ \exp \ps{ O\pfrac{kd \cdot \ln q}{\eps}} }.
\end{align*}
Assuming $k,d, q = O(1)$ are all constants, it follows that there exists some constant $\eta = \eta(k,d,q)$ such that the above is bounded by $\poly(n)$ for
\[
    \eps\geq\frac{\eta}{\log\log n}.
\]
\end{proof}

\subsection{Setup}
Our main structural lemma, \cref{lem:small-bad-sequence}, shows that if the instance is $\eps$-frustrated, then from every string there is a bounded-length overlapping sequence reaching a BAD string.

To this end, we use a probabilistic argument, showing that a sufficiently long randomly chosen sequence of terms $P_{i_1},\ldots,P_{i_T}$ satisfies the following:
\begin{itemize}
    \item at some time $t$, the set $S_{t-1}$ contains a $P_{i_t}$-BAD string;
    \item there exists a relatively short overlapping subsequence that leads $x$ to a BAD string, say $w$.
\end{itemize}
Since all the terms act locally, being $P$-BAD depends only on the few qudits on which $P$ acts. Using a standard lightcone argument, we can trace backward the terms that influenced these qudits throughout our sequence. Only these terms are eventually relevant for modifying $x$ into a $P$-BAD string. It follows that we should bound the length of the longest \emph{overlapping subsequence} in order to bound the graph distance between $x$ and $w$.

\begin{definition}
    Let \(\vec Q=(Q_1,\ldots,Q_t)\) be a sequence of terms in \(\mathcal P\). The \emph{overlapping subsequence} of \(\vec Q\), denoted by \(\overlapseq(\vec Q)\), is the longest ordered subsequence \( Q_{i_1}, \ldots, Q_{i_t} \) such that $i_1 = 1$ and for every $2 \leq j \leq t$, $Q_{i_j}$ overlaps $Q_{i_k}$ for some $k < j$.
\end{definition}
Formally, it can be constructed as follows:

\begin{algorithm}[H]
\label{alg:overlapping-subseq}
\KwIn{A sequence of terms $\vec Q = (Q_1,\ldots,Q_t)$.}
\KwOut{An overlapping subsequence $\vec W \subseteq \vec Q$.}
$\vec W \gets (Q_1)$\;
\For{$j = 2,\ldots,t$}{
  \If{$\exists\, Q \in \vec W$ such that $\supp(Q_j) \cap \supp(Q) \neq \emptyset$}{
    append $Q_j$ to $\vec W$\;
  }
}
\Return $\vec W$\;
\caption{$\overlapseq$}
\end{algorithm}

Without loss of generality, we assume that $m\ge n$.\footnote{After removing unused qudits, we have $n\le km$. We can therefore duplicate all terms at most $k$ times to obtain $m\ge n$. This leaves the averaged Hamiltonian, and hence its frustration, unchanged, while increasing the degree by only a constant factor.}
Throughout, we consider the parameters
\begin{align}\label{eq:def parameters T elleps}
    T\defeq \left\lceil \frac{16m \ln q}{\eps}\right\rceil
    \qquad\text{and}\qquad
    \ell_\eps \defeq \left\lceil \frac{2}{\eps} q^{16kd/\eps} \right\rceil.
\end{align}
Fix an \(\eps\)-frustrated instance \(\cP\).
Throughout this section, we use the following notations. 
Let $x\in \Sigma^n$ be arbitrary.
Consider $T$-many indices, which we will later show how to choose:
\[
i_1,\ldots, i_T \in [m].
\]
Fixing these choices, we define the following quantities.

\begin{itemize}

    \item For every \(t\leq T\), denote by $D_t$ the length of the overlapping subsequence of the reverse sequence $(P_{i_t},\ldots,P_{i_1})$, namely
    \[
    D_t \defeqq \abs{\overlapseq (P_{i_t},P_{i_{t-1}},\ldots,P_{i_1}) }.
    \]

    \item 
    %Denote the subset state corresponding to applying the sequence $P_{i_1},\ldots, P_{i_T}$ on $x$.
    Starting from \(S_0 \defeq \set{x}\), for every $t \le T$ define the set $S_t \subseteq \Sigma^n$ as
    \[
    S_t \defeq \supp\ps{P_{i_t} \ket{S_{t-1}}}.
    \]
    We identify every set $S_t$ with its uniform subset state $\ket{S_t}$.

    \item 
    For every \(t\leq T\), define the $t$-th step frustration as
    \[
    \gamma_t\defeq 1 -\norm{P_{i_t} \ket{S_{t-1}}}_2^2.
    \]
\end{itemize}

Our goal is to show that, for some choice of indices \(i_1,\ldots,i_T\), the sets $S_1,\ldots,S_T$ steadily increase in size. In turn, $S_T$ would be too large unless a BAD string appears earlier.
To this end, let $\cI$ be an arbitrary distribution on $[m]^T$, and suppose the indices are sampled from it:
\[
(i_1,\ldots, i_T) \sim \cI.
\]
We now relate the expected size of the final set, $S_T$, in this random process, to the expected frustration at every intermediate step.

\begin{lemma}\label{lem:small-bad-sequence random}
Let $\cI \subseteq [m]^T$ be an arbitrary distribution, and let $(i_1,\ldots, i_T) \sim \cI$.

Assume that, with probability one, for every $t$ the set $S_{t-1}$ contains no $P_{i_t}$-BAD string.
Then, for every \(x\in\Sigma^n\),
\[
\Expect{(i_1,\ldots, i_T) \sim \cI}{\abs{S_T}} \ge
\exp \left(\frac 1 4 \sum_{t=1}^T \Expect{(i_1,\ldots, i_T) \sim \cI}{\gamma_t} \right).
\]
\end{lemma}
\begin{proof}
By the hypothesis, for every $t$, the set $S_{t-1}$ contains no $P_{i_t}$-BAD strings, and hence we can apply \cref{thm:ag19 frustration expansion}. Using the inequality $1+\delta\geq e^{\delta/2}$ for all $\delta\in[0,1]$, we conclude
\[
    |S_t|\geq\left( 1+\frac{\gamma_t}{2} \right)|S_{t-1}|\geq e^{\gamma_t/4}|S_{t-1}|.
\]
Applying this argument inductively, taking expectations and applying Jensen's inequality, we obtain
\begin{align*}
\Expect{(i_1,\ldots, i_T) \sim \cI}{|S_T|}
\ge  \Ex_{(i_1,\ldots, i_T) \sim \cI}  \exp \left( \frac14 \sum_{t=1}^T \gamma_t\right)
\ge \exp \left(\frac 1 4 \sum_{t=1}^T \Expect{(i_1,\ldots, i_T) \sim \cI}{ \gamma_t} \right).
\end{align*}
\end{proof}

The role of the preceding lemma is to force a BAD string at an intermediate time whenever the chosen distribution makes $\sum_t\Ex[\gamma_t]$ sufficiently large. Indeed, if every fixation of indices and time $t$ had $S_{t-1}$ containing no $P_{i_t}$-BAD string, then the lemma would give $\Ex[|S_T|]>|\Sigma^n|$, contradicting the trivial bound $S_T\subseteq\Sigma^n$. Hence, for the distributions considered below, some indices and time $t$ satisfy that $S_{t-1}$ contains a $P_{i_t}$-BAD string. It remains to bound the distance from $x$ to such a string, which reduces to bounding $D_t$.

In \cref{sec:uniform index,sec:chopped-seq} we show two bounds for two distributions $\cI$.

\subsection{Uniformly Random Sequence}\label{sec:uniform index}
In this subsection we analyze the above process when $\cI$ is the uniform distribution on $[m]^T$.
First, observe that the expected frustration of a randomly chosen term is at least the frustration of $\cP$:
\begin{lemma}\label{thm:uniform frustration}
    For every $t\leq T$, if $i_t \inr [m]$ is selected uniformly at random, we have that
    \[
    \Expect{i_t \inr [m]}{\gamma_t} \ge \eps.
    \]
\end{lemma}
\begin{proof}
    Fix \(i_1,\ldots,i_{t-1}\) arbitrarily, and let \(\ket{S_{t-1}'}\) be the corresponding fixed state. Then, since $\cP$ is $\eps$-frustrated,
    \[
    \cExpect{i_t\inr [m]}{i_1,\ldots,i_{t-1}}{\gamma_t}
    = \frac{1}{m}\sum_{i_t \in [m]}\left(1- \norm{P_{i_t} \ket{S_{t-1}'}}^2\right)
    = 1 - \bra{ S'_{t-1}} \frac{1}{m} \sum_{i_t\in [m]}  P_{i_t} \ket{S'_{t-1}}
    \geq\eps.
    \]
The lemma follows since the above holds for every fixed choice of \(i_1,\ldots,i_{t-1}\).
\end{proof}

We observe that \cref{lem:small-bad-sequence random,thm:uniform frustration} already imply that some choice of indices must encounter a $P_{i_t}$-BAD string in $S_{t-1}$ at some time $t$. Indeed, if this never occurred, then for $T\ge \tfrac{8n}{\eps}\ln q$,
\[
\Expect{i_1,\ldots,i_T\inr[m]}{|S_T|}
\ge \exp\left(\frac14\sum_{t=1}^T\Expect{i_1,\ldots,i_T\inr[m]}{\gamma_t}\right)
\ge \exp\left(\frac{T\eps}{4}\right)
\ge q^{2n}.
\]
This contradicts the trivial bound $|S_T|\le q^n$.

% Indeed, otherwise, the support would continue expanding beyond the size $q^n$ of the entire universe, contradicting the guaranteed existence of a BAD string.

However, the issue is that such a sequence does not guarantee any \emph{bound} on the maximal path from $x$ to the BAD string, besides the trivial bound $D_t \le T$.
We address this issue next.

\subsection{The Chopped Sequence}\label{sec:chopped-seq}
In this subsection, we still pick the indices uniformly at random, but we remove some of them.
Informally, we consider the ``chopped'' random sequence of $(P_{i_1}, \ldots, P_{i_T})$, where we remove terms $P_{i_t}$ that can lead to long overlapping subsequences. We then analyze this random process and show that eventually the modified random sequence both (a) expands sufficiently and contains a BAD string, and, more importantly, (b) gives a bound on the maximal distance between $x$ and a BAD string.

Formally, define the new ``chopped'' random sequence as follows. 
The indices are chosen uniformly at random:
\[
(i_1,\ldots, i_T) \inr [m]^T.
\]
We now define the sequence $\wt P_{i_1}, \ldots, \wt P_{i_T}$, which is related to the random sequence $P_{i_1}, \ldots P_{i_T}$. For every $t$ we define
\[
\widetilde P_{i_t}\defeq
\begin{cases}
P_{i_t}, & D_t\leq\ell_\eps,\\
I, & D_t>\ell_\eps.
\end{cases}
\]
Correspondingly, we define the support states as $\wt S_0 = \set{x}$, and for all $t\geq 1$,
\[
\wt S_t \defeq \supp \ps{ \wt{P}_{i_t} \ket{\wt S_{t-1}} }.
\]
The corresponding frustration of the $t$-th step in the chopped random process is defined as 
\[
\widetilde \gamma_t \defeq 1-\norm{\widetilde P_{i_t}\ket{\wt S_{t-1}}}_2^2.
\]
We now relate the expected frustration of the chopped sequence to the chopping probability.

\begin{lemma}\label{lem:censored-frustration}
For every \(t\leq T\),
\[
\Expect{}{\widetilde\gamma_t} \geq \eps - \Prob{}{D_t > \ell_\eps}.
\]
\end{lemma}
\begin{proof}
Define the ``hybrid'' frustration random variable, of the original term $P_{i_t}$ with respect to the current state of the chopped process:
\[
    \overline{\gamma}_t
    \coloneqq
    1- \norm {
        P_{i_t} \ket{ \widetilde S_{t-1}}
      }_2^2.
\]
Note that if \(\widetilde P_{i_t}=I\), then by definition there is no frustration:
\[
\wt \gamma_t = 1 - \norm{I \cdot \ket{\wt S_{t-1}}}^2 = 0,
\]
and since that happens whenever $D_t > \ell_\eps$, we can write
\[
\widetilde\gamma_t = \ol \gamma_t\mathbf 1_{D_t\leq\ell_\eps}.
\]
Then, since the frustration is bounded by \(\ol \gamma_t\leq1\),
\begin{align*}
\Expect{}{\widetilde \gamma_t}
&=\Expect{}{ \ol \gamma_t\mathbf 1_{D_t \leq \ell_\eps}}\\
&=\Expect{}{\ol \gamma_t} - \Expect{}{\ol \gamma_t \mathbf 1_{D_t > \ell_\eps}}\\
&\geq \eps - 1\cdot \Prob{}{D_t>\ell_\eps },
\end{align*}
where the proof that $\Expect{}{\ol \gamma_t} \ge \eps$ follows similarly to \cref{thm:uniform frustration}.
\end{proof}

Therefore, if we can bound the probability of having a long overlapping subsequence by, say, $\eps/2$, we can guarantee both (a) the expected chopped frustration is still high, $\Expect{}{\wt{\gamma}_t} \ge \eps/2$, which is necessary for the gradual expansion of the sets $\wt S_t$, and (b) a bound on the maximal path between $x$ and a BAD string.

We now show that, for every fixed step, its overlapping subsequence is relatively short with high probability:
\begin{lemma}\label{lem:small-backward-sequence}
For every \(t\leq T\),
\[
\Pr[D_t>\ell_\eps]\leq\frac{\eps}{2}.
\]
\end{lemma}

\begin{proof}
Fix \(t\), and consider the construction of \(\overlapseq(P_{i_t},\ldots,P_{i_1})\) at step $j\in[t-1,1]$. Suppose that \(r\) terms have already been retained. Therefore at most \(kdr\) terms of \(\cP\) overlap at least one of them, and since \(i_j\) is uniform and independent of the previously fixed indices \(i_t,\ldots,i_{j+1}\), it follows that the current term is retained with probability at most
\[
\cProb{i_j}{i_t,\ldots,i_{j+1}}{\text{\(P_{i_j}\) is retained}}\leq\frac{kdr}{m}.
\]
Thus, the expected number of retained terms grows by a factor of at most \((1+\tfrac{kd}{m})\) at every step. Starting from the single root term, and using $t-1\leq T-1<\frac{16m\ln q}{\eps}$, we obtain
\[
\Expect{}{D_t}
= \Expect{}{\abs{\overlapseq(P_{i_t}, \ldots, P_{i_1})}}
\leq \left(1+\frac{kd}{m}\right)^{t-1}\leq e^{kd(t-1)/m} \leq q^{16kd/\eps}.
\]
Therefore, by Markov's inequality,
\[
\Prob{}{D_t>\ell_\eps}\leq\frac{\Expect{}{D_t}}{\ell_\eps}\leq\frac{\eps}{2}.
\]
\end{proof}

Plugging the above bounds into \cref{lem:small-bad-sequence random}, we obtain:

\begin{lemma}\label{thm:expected set of choppset set}
    Let $x\in \Sigma^n$ be arbitrary.
    Suppose that, with probability one, for every $t$ the set $\wt S_{t-1}$ contains no $\wt P_{i_t}$-BAD string.
    Then,
    \[
    \Ex [|\wt S_T|] \ge q^{2n}.
    \]
\end{lemma}

\begin{proof}
Combining \cref{lem:censored-frustration,lem:small-backward-sequence}, we have that $\Expect{}{\wt \gamma_t} \ge \eps /2$.
Let $\wt{\cI}$ denote the distribution on chopped sequences induced by the uniformly random indices. Applying \cref{lem:small-bad-sequence random} to $\wt{\cI}$ and the states $\wt{S}_t$ gives
\begin{equation*}
\Expect{}{\abs{ \wt S_T }}
\ge
\exp \ps {\frac 1 4 \sum_{t=1}^T \Expect{}{\widetilde \gamma_t} }
\geq \exp \ps{ \frac{ \eps T}{8} }
\geq q^{2n}. \qedhere 
\end{equation*}
\end{proof}

At a high level, the preceding expectation bound forces some realization of the chopped process to encounter a BAD string.
Moreover, by design, the chopped process induces a relatively short path from any string to a BAD one, allowing us to conclude:

\begin{lemma}\label{lem:small-bad-sequence}
For every \(x\in\Sigma^n\), there exist a term \(P\in\cP\) and a rooted overlapping sequence $(P,Q_1,\ldots,Q_\ell)$, of length $\ell < \ell_\eps$, such that $\supp(Q_1\cdots Q_\ell\ket{x})$ contains a $P$-BAD string.
\end{lemma}

\begin{proof}
Fix $x\in\Sigma^n$. Suppose, toward a contradiction, that no fixing of the indices encounters a $\wt{P}_{i_t}$-BAD string in $\wt{S}_{t-1}$ at any time $t$.
Then the hypothesis of \cref{thm:expected set of choppset set} holds for every choice of indices, and hence
\[
\Ex[|\wt{S}_T|]\ge q^{2n}.
\]
This contradicts the trivial bound $|\wt{S}_T|\le q^n$. Therefore, for some $i_1,\ldots, i_T$ and some time $t$, the set $\wt{S}_{t-1}$ contains a $\wt{P}_{i_t}$-BAD string. Necessarily $\wt{P}_{i_t}=P_{i_t}$, since no string is $I$-BAD. By the definition of the chopped sequence, this implies $D_t\le\ell_\eps$.

Finally, form the backward lightcone of the \emph{chopped} sequence $(\wt{P}_{i_t},\ldots,\wt{P}_{i_1})$, discarding identity operators. Denote this lightcone, with its root listed first, by $(P_{i_t},Q_1,\ldots,Q_\ell)$. Replacing terms by the identity can only remove dependencies, so this lightcone is contained in $\overlapseq(P_{i_t},\ldots,P_{i_1})$ and contains at most \(D_t\) terms (including its root \(P_{i_t}\)).

We now show that the non-root terms of this lightcone still lead from \(x\) to a \(P_{i_t}\)-BAD string. To this end, choose a $P_{i_t}$-BAD string $y\in\wt S_{t-1}$. Since the projectors are entrywise non-negative,
\begin{align}\label{eq:nonnegativity}
    0<\bra{y}\wt{P}_{i_{t-1}}\cdots\wt{P}_{i_1}\ket{x}.
\end{align}
Let $z\defeq y|_{\supp(P_{i_t})}$ be the local BAD pattern witnessed by $y$. Consider any term of the chopped sequence $\wt P_{i_j}$ that does not appear among $Q_1,\ldots,Q_\ell$. By construction of the backward lightcone, this term shares no qudit with $P_{i_t}$ or with any retained term $Q_r$ applied after it. Thus, any change caused by this term cannot propagate to the final symbols on $\supp(P_{i_t})$. Removing all such terms from the product in \cref{eq:nonnegativity} therefore gives some $y'\in\Sigma^n$ such that
\[
0<\bra{y'}Q_1\cdots Q_\ell\ket{x},
\qquad
y'|_{\supp(P_{i_t})}=z.
\]
Hence $y'\in\supp(Q_1\cdots Q_\ell\ket{x})$ is also $P_{i_t}$-BAD, which has distance $\ell \le D_t -1 \le \ell_\eps$ from $x$, as required.
\end{proof}

\newcommand{\parent}{\mathsf{parent}}
\section{An Instance For Which Exponentially Many Overlapping Sequences Are $\delta$-valid}\label{sec:example delta valid cones}
In this section, we show for every $\delta > 0$, an explicit $O(\delta)$ frustrated instance $\cP\in\ProjUSLHinstance$ which exhibits the following phenomenon:
a large subset of the rooted overlapping sequences are sequentially $\delta$-frustrated with respect to some single witness $x\in \Sigma^n$.
Thus, restricting the enumeration only to such a subset does not improve the worst-case dependence on the sequence length.

To have some sense of the parameter, the reader can think of any $\delta\le 1-\sqrt{2/3}$, for which we may take constant-size alphabet dimension $q=3$ (and the set $A=\set{0,1}$ downstream). In particular, for all sufficiently large $n$, the construction applies with constant alphabet size to any sub-constant frustration $\delta=o(1)$.

\subsection{Construction}
Let $\delta\in(0,1)$ be arbitrary. 
Fix an integer $k\ge2$.
Let the alphabet size be $q\defeqq \ceil{\tfrac{2}{(1-\delta)^2}}$, and set $\Sigma = \set{0, \dots, q-1}$. We fix the set
\[
A\defeq \set{0, \dots, \lfloor(1-\delta)q\rfloor -1 }\subseteq \Sigma.
\]
We begin by describing the interaction graph.
Let $T\ge3$ be an integer TBD.
Consider the complete depth-$T$ $k$-ary tree with root $\emptyset$, whose vertex set is defined as
\[
V\defeqq \set{\emptyset} \bigcup_{t=1}^{T-1}[k]^t. 
\]
Internal nodes at depth $t < T-1$ are denoted by strings $v\in[k]^t$, and their children are denoted by $v1,\dots,vk$. Each of them is naturally associated with their parent, which is defined as
\[
\parent(vj) \defeq v.
\]
We associate one qudit over $\Sigma$ per node, so $[n]=V$, and therefore we choose $T$ such that
\[
n = \frac{k^T-1}{k-1}.
\]
Next, we describe the Hamiltonian's terms.
The instance is defined as  $\cP = \set{P^{(v)}}_{v\in[n]}$ where every term is defined as
\begin{align}\label{eq:ex construction}
P^{(v)} \defeq \ketbra{T_v}{T_v},
\qquad \text{ where } \qquad
\ket{T_v} \defeq
\begin{cases}
\ket{A}_v\bigotimes_{j=1}^{k}\ket{\Sigma}_{vj}, & v\text{ internal},\\
\ket{A}_v, & v\text{ leaf}.
\end{cases}
\end{align}
\Cref{fig:ex-construction} illustrates the construction.

\subsection{Validity of the Instance}
First, observe that the instance is valid:
\begin{lemma}
    $\cP \in \ProjUSLHinstancee_q(k+1, d=2)$.
\end{lemma}
\begin{proof}
    By construction, each $P^{(v)}$ is either $1$-local or $(k+1)$-local.
    Additionally, observe that each non-root qudit is acted upon only by $P^{(v)}$ and $P^{(\parent(v))}$, so each qudit is acted upon by at most $d=2$ terms. 
    Finally, each $\ket{T_v}$ is a singleton uniform subset state, corresponding to either the set $A \times (\Sigma)^k$ or $A$.
\end{proof}

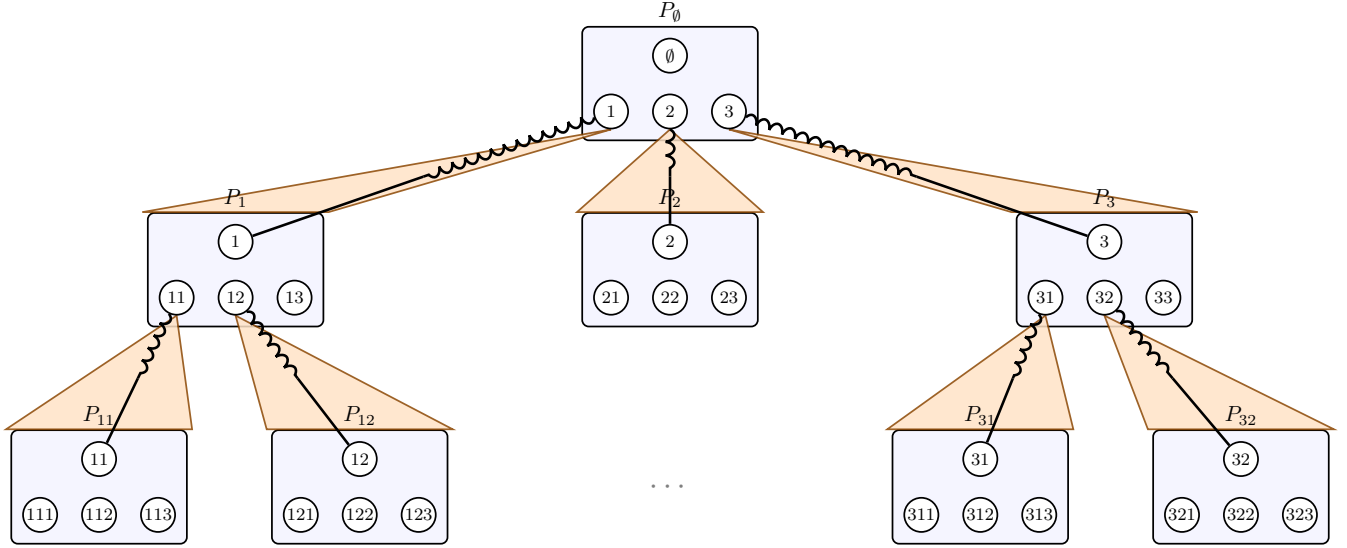
\begin{figure}[t]
  \centering
  \resizebox{\textwidth}{!}{%
  \begin{tikzpicture}[
    qubit/.style={circle,draw,thick,minimum size=5.5mm,inner sep=0pt,font=\scriptsize,fill=white},
    block/.style={draw,thick,rounded corners=3pt,fill=blue!4,inner sep=5pt},
    blocktitle/.style={font=\small\bfseries,inner sep=1pt},
    conefill/.style={fill=orange!22,draw=orange!55!black,thick,opacity=0.85},
    solidline/.style={very thick,black},
  ]
  \def\drawblock#1#2#3#4{%
    \node[qubit] (#3-top) at (#1, #2+0.45) {$#4$};
    \node[qubit] (#3-c1) at (#1-0.95, #2-0.45) {$#4 1$};
    \node[qubit] (#3-c2) at (#1,      #2-0.45) {$#4 2$};
    \node[qubit] (#3-c3) at (#1+0.95, #2-0.45) {$#4 3$};
    \node[blocktitle] at (#1, #2+1.15) {$P_{#4}$};
    \begin{scope}[on background layer]
      \node[block,fit=(#3-top)(#3-c1)(#3-c3)] (#3-blk) {};
    \end{scope}
  }
  % root (with relabeled children 1,2,3)
  \node[qubit] (R-top) at (0, 6.5+0.45) {$\emptyset$};
  \node[qubit] (R-c1) at (-0.95, 6.5-0.45) {$1$};
  \node[qubit] (R-c2) at ( 0,    6.5-0.45) {$2$};
  \node[qubit] (R-c3) at ( 0.95, 6.5-0.45) {$3$};
  \node[blocktitle] at (0, 6.5+1.15) {$P_{\emptyset}$};
  \begin{scope}[on background layer]
    \node[block,fit=(R-top)(R-c1)(R-c3)] (R-blk) {};
  \end{scope}
  % layer 2
  \drawblock{-7}{3.5}{A}{1}
  \drawblock{0}{3.5}{B}{2}
  \drawblock{7}{3.5}{C}{3}
  % layer 1 (expand two grandchildren each under P_1 and P_3)
  \drawblock{-9.2}{0}{A1}{11}
  \drawblock{-5.0}{0}{A2}{12}
  \drawblock{5.0}{0}{C1}{31}
  \drawblock{9.2}{0}{C2}{32}
  \node[font=\Large\bfseries,gray] at ( 0,0) {$\cdots$};
  % shaped cones from parent's i-th child-slot to child block (its left and right sides)
  \foreach \pa/\ch/\i in {R/A/1, R/B/2, R/C/3, A/A1/1, A/A2/2, C/C1/1, C/C2/2}{
    \begin{scope}[on background layer]
      \fill[conefill] (\pa-c\i.south)
        -- ([xshift=-2pt]\ch-blk.north west)
        -- ([xshift=2pt]\ch-blk.north east)
        -- cycle;
    \end{scope}
  }
  % shared-register lines: split at midpoint — decorated half toward parent, plain solid toward child
  \foreach \pa/\ch/\i in {R/A/1, R/B/2, R/C/3, A/A1/1, A/A2/2, C/C1/1, C/C2/2}{
    \path (\pa-c\i) -- (\ch-top) coordinate[pos=0.5] (mid-\pa-\ch);
    % parent half: snake/coil
    \draw[very thick,decorate,decoration={coil,aspect=0.5,segment length=2.2mm,amplitude=0.7mm}]
      (\pa-c\i) -- (mid-\pa-\ch);
    % child half: plain solid
    \draw[very thick] (mid-\pa-\ch) -- (\ch-top);
  }
  \end{tikzpicture}%
  }
  \caption{The instance $\mathcal P$ for $k=\ell=3$. Circles denote qudits, labeled by their register name $v$; each block groups the $k+1$ qudits acted on by a single term $P^{(v)}$, with the ``parent'' register $v$ on top and the ``child'' registers $v1,\dots,vk$ beneath. Every register $v$ is shared by exactly two terms, $P^{(v)}$ and $P^{(\parent(v))}$. The line and shaded cone mark this overlap, connecting the register's appearance as the parent's $i$-th child slot to its appearance as the child block's top register. The two ends are styled differently to mark the asymmetric action on the shared register: the parent applies $\ketbraa{\Sigma}_v$ (coiled end) and the child applies $\ketbraa{A}_v$ (straight end).}
  \label{fig:ex-construction}
\end{figure}

Second, we show that the instance is frustrated, thereby a NO case.

\begin{lemma}\label{lem:energy}
$\cP$ is $\delta/8k$-frustrated.
\end{lemma}
\begin{proof}
Let $\ket \psi$ be an arbitrary state, and consider its energy:
Recall that $|V|=n$, and hence
\begin{align}\label{eq:frustration}
    \bra \psi \frac{1}{n} \sum_{v\in V} P^{(v)} \ket \psi
= \frac1n \sum_{v\in V} \norm{P^{(v)}\ket\psi}^2
\le \frac1n \sum_{v\in V} \norm{P^{(v)}\ket\psi}.
\end{align}
Fix a pair of parent-child registers $(v,v1)$, and we show how to bound $\norm{P^{(v)} \ket \psi} + \norm{P^{(v1)} \ket \psi}$. Later we'll show how to pair all registers in $V$.

Each pair of terms $P^{(v)},P^{(v1)}$ acts simultaneously only on register $v1$. Since $\ketbraa\phi\preceq I$ for any unit-state $\ket\phi$, we have that each projector is smaller in the Löwner partial order than the following projectors:
\begin{equation}\label{eq:Pv lowener}
    \begin{split}
        P^{(v)} &\preceq \ketbra{\Sigma}{\Sigma}_{v1}\otimes I_{v,v2,\ldots,vk} = \ketbraa{\Sigma}_{v1} \\
P^{(v1)} &\preceq \ketbra{A}{A}_{v1}\otimes I_{v11,\ldots, v1k} = \ketbraa{A}_{v1}.
    \end{split}
\end{equation}
Therefore, using the inequality $\sqrt x + \sqrt y \le \sqrt{2(x+y)}$ we conclude that
\begin{align*}
    \norm{P^{(v)} \ket \psi} + \norm{P^{(v1)} \ket \psi}
    &=
    \sqrt{ \bra\psi P^{(v)}\ket\psi } + \sqrt{ \bra\psi P^{(v1)}\ket\psi} \\
    &\le
    \sqrt{2\ps{\bra \psi P^{(v)} \ket \psi + \bra \psi P^{(v1)} \ket \psi }} \\
    &\le
    \sqrt{ 2\bra{\psi} \ps{ \ketbraa{\Sigma}_{v1} + \ketbraa{A}_{v1}} \ket \psi }\\
    &\defeq
    \sqrt{2\bra\psi M^{(v1)}\ket\psi}.
\end{align*}
In \cref{thm:spectra of Mv1}, we show that the operator $M^{(v1)}$  has maximal eigenvalue $\lambda_{\max}(M^{(v1)}) \le 2- \delta/2$. Using the inequality $\sqrt{1-x} \le 1-x/2$ for $x\in[0,1]$, we upper bound the above by 
\[
\norm{P^{(v)} \ket \psi} + \norm{P^{(v1)} \ket \psi} \le \sqrt{4-\delta} = 2 \sqrt{1 - \delta/4} \le 2-\delta/4.
\]
Now, let $W \subseteq V$ be the set of $(T-2)$-depth nodes, namely the internals whose children are all leaves, so $\abs{W}=k^{T-2}$. Pair each $v\in W$ with its first child $v1$. 
Therefore, using $T\ge3$ and $k\ge2$,
\[
\frac{\abs{W}}{n} = k^{T-2} \frac{k-1}{k^T-1} \ge \frac{k-1}{k^2} \ge \frac{1}{2k}.
\]
Finally, we use the preceding pairwise bound and bound the remaining $n-2|W|$ projection norms trivially by $1$. Plugging it into \cref{eq:frustration}, we conclude that
\[
\bra\psi\frac1n\sum_{v\in V}P^{(v)}\ket\psi
\le \frac{n-2\abs{W}}{n}\cdot 1 +\frac{\abs{W} }{n} \cdot(2-\delta/4)
= 1-\frac{\abs{W}\delta/4}{n}
\le 1-\frac{\delta}{8k}.
\]
\end{proof}

\begin{lemma}\label{thm:spectra of Mv1}
    For every internal node $u \in V$, define the operator
    \[
    M^{(u)} \defeq \ketbraa{A}_{u} + \ketbraa{\Sigma}_{u}.
    \]
    Then
    $\lambda_{\max}(M^{(u)}) \le 2 -\delta/2$.
\end{lemma}
\begin{proof}
    First, observe that the relative size of $A \subseteq \Sigma$ is bounded by
    \[
    \mu \defeq \braket{A}{\Sigma} = \sqrt{\abs{A}/ \abs{\Sigma}} \le \sqrt{1-\delta} \le 1-\delta/2.
    \]    
    Now observe that
    \begin{align*}
    &M^{(u)} \ps{\ket A+\ket{\Sigma}}=(1+\mu)\ps{\ket A+\ket{\Sigma}},\\
    &M^{(u)} \ps{\ket A-\ket{\Sigma}}=(1-\mu)\ps{\ket A-\ket{\Sigma}}.
    \end{align*}
    On the other hand, $M^{(u)}\ket\chi=0$ for every state $\ket \chi$ orthogonal to $\textrm{span}\set{\ket A, \ket{\Sigma}}$. Hence $M^{(u)}$'s eigenvalues are $\set{1\pm\mu, 0}$, and since $\mu > 0$, the maximal eigenvalue is $\lambda_{\max}(M^{(u)}) \le 1 + \mu \le 2-\delta/2$.
\end{proof}

\subsection{Doubly Exponentially Many Sequentially Frustrated Sequences}
\newcommand{\nicesequences}{\mathsf{NiceOverlapSeq}}
\begin{definition}\label{def:sequentially frustrated}
Let $x\in \Sigma^n$. A sequence of terms $(Q_1,\ldots,Q_t)$ is \emph{sequentially $\delta$-frustrated with respect to $x$} if, for every $j\in[t]$:
\[
    \norm{Q_j \cdot \wh{\ps{ Q_{j-1}\cdots Q_1\ket{x} }}}_2^2 \le 1 - \delta.
\]
\end{definition}

Fix a term $P^{(u)}\in\mathcal P$ and a length bound $\ell\ge1$.
\begin{definition}
    A sequence $(P^{(v_1)},\ldots,P^{(v_t)})$ is \emph{nice below $P^{(u)}$} if its nodes form an ordered top-down subtree below $u$. Equivalently, for every $j\in[t]$, the parent of $v_j$ is either $u$, or some $v_{j'}$ with $j'<j$.
\end{definition}
The set of such sequences having at most $\ell$ applied terms is denoted by
\[
\nicesequences_{P^{(u)},\ell} = \sett{(P^{(v_1)},\ldots,P^{(v_t)})}{0\le t\le\ell\text{ and the sequence is nice below $P^{(u)}$}}.
\]

We first observe that the number of nice sequences is large.
\begin{lemma}\label{thm:nice overlap min size}
Let $\ell_\varepsilon\ge2$. For every node $v$ with at least $\ell_\varepsilon-1$ levels below it,
\[
\abs{\nicesequences_{P^{(v)},\ell_\varepsilon}}
\ge 
2^{\ell_\varepsilon/(2k)}.
\]
\end{lemma}

\begin{proof}
Let $h$ be maximal such that the complete depth-$h$ subtree rooted at $v$ contains at most $\ell_\varepsilon$ nodes, namely
\[
1+k+\cdots+k^h\le\ell_\varepsilon.
\]
Every subset of the $k^h$ leaves determines a distinct subtree of size at most $\ell_\varepsilon$, and hence a distinct nice overlapping sequence.
By the maximality of $h$ we have $\ell_\varepsilon \le 2k^{h+1},$ and hence $k^h>\ell_\varepsilon/(2k)$. Therefore, the number of nice overlapping sequences of length at most $\ell_\varepsilon$ is at least
\[
2^{k^h}
\ge
2^{\ell_\varepsilon/(2k)}.
\]
\end{proof}

Next, we prove that all nice overlapping sequences are $\delta$-frustrated with respect to some witness $x \in \Sigma^n$.

\begin{lemma}\label{lem:frustrated-any-root}
For every $P\in\cP$ and length bound $\ell>0$, if $(Q_1,\ldots,Q_t)\in\nicesequences_{P,\ell}$, then the reversed sequence $(Q_t,\ldots,Q_1)$ is sequentially $\delta$-frustrated with respect to $x=0^n$.
\end{lemma}

\begin{proof}
Fix $P\in\cP$, a length bound $\ell>0$, and a nice sequence $(Q_1, \ldots ,Q_t)\in \nicesequences_{P,\ell}$.
Fix an arbitrary $j \in [t-1]$, and let
\[
\ket\Psi\defeqq Q_{j+1}\cdots Q_t\ket{x}.
\]
We show that
\begin{align}\label{eq:target}
\norm{Q_j\ket{\wh\Psi}}_2^2 \le 1-\delta.
\end{align}
Suppose $P^{(v)} = Q_j$.
The key observation is that, immediately before applying $P^{(v)}$, register $v$ has never been acted upon.
Indeed, the only terms acting on register $v$ are $P^{(v)}$ and, if $v$ is non-root, $P^{(\parent(v))}$. Since the sequence is listed top down, $P^{(\parent(v))}$ appears before $P^{(v)}$ and is therefore applied later; if $P^{(\parent(v))}=P$, then it is not part of the sequence. Hence none of the terms $Q_{j+1},\ldots,Q_t$ acts on register $v$, and therefore
\[
\ket{\wh\Psi}_v=\ket0_v.
\]
On register $v$, the projector $P^{(v)}$ acts as $(P^{(v)})_v=\ketbraa A_v.$
Recall that $P^{(v)}\preceq \ketbraa A_v\otimes I$ (\cref{eq:Pv lowener}). Therefore,
\begin{align*}
    \norm{P^{(v)}\ket{\wh\Psi}}_2^2
&=
\bra{\wh\Psi}P^{(v)}\ket{\wh\Psi} \\
&\le 
\bra{\wh\Psi}(\ketbraa A_v\otimes I)\ket{\wh\Psi} \\
&\le \abs{\braket{A}{0}}^2=\frac1{|A|}.
\end{align*}
Then by our choice $\abs A=\floor{(1-\delta)q}$ and $q\ge 2/(1-\delta)^2$, we have
\begin{align*}
    \abs A
&\ge (1-\delta)q-1
\ge \frac{2}{1-\delta}-1
\ge \frac{1}{1-\delta}.
\end{align*}
and hence
\[
\norm{P^{(v)}\ket{\wh\Psi}}_2^2
\le \frac{1}{\abs A}
\le 1-\delta.
\]
Since $P^{(v)}$ was arbitrary, every term application is sequentially
$\delta$-frustrated, proving \cref{lem:frustrated-any-root}.
\end{proof}

\section*{Acknowledgments}
We thank John Bostanci and Uma Girish for collaboration in the early stages of this project, and Alex Grilo and Dorit Aharonov for helpful discussions. This work began during a visit to Henry Yuen's group, and we thank him for his hospitality.

\paragraph{AI Disclosure.}
This research was conducted by the human authors with assistance from AI tools. Claude and ChatGPT were used as research and writing assistants for brainstorming and proof exploration. ChatGPT was used during the development of the improvement from the initial \(\beta=\Omega(1/\sqrt{\log\log n})\) algorithm to the final \(\beta=\Omega(1/\log\log n)\) result, and Claude assisted with exploratory ideas related to \cref{sec:example delta valid cones}.
The authors verified the correctness and originality of all content including references.

\printbibliography

@inproceedings{AG19,
  title={Stoquastic PCP vs. randomness},
  author={Aharonov, Dorit and Grilo, Alex Bredariol},
  booktitle={2019 IEEE 60th Annual Symposium on Foundations of Computer Science (FOCS)},
  pages={1000--1023},
  year={2019},
  organization={IEEE}
}

@misc{BBT06,
      title={Merlin-Arthur Games and Stoquastic Complexity}, 
      author={Sergey Bravyi and Arvid J. Bessen and Barbara M. Terhal},
      year={2006},
      eprint={quant-ph/0611021},
      archivePrefix={arXiv},
      primaryClass={quant-ph},
      url={https://arxiv.org/abs/quant-ph/0611021}, 
}

@article{KVM02,
  title={Graph nonisomorphism has subexponential size proofs unless the polynomial-time hierarchy collapses},
  author={Klivans, Adam R. and van Melkebeek, Dieter},
  journal={SIAM Journal on Computing},
  volume={31},
  number={5},
  pages={1501--1526},
  year={2002},
  publisher={SIAM}
}

@article{IKW02,
  title={In search of an easy witness: Exponential time vs. probabilistic polynomial time},
  author={Impagliazzo, Russell and Kabanets, Valentine and Wigderson, Avi},
  journal={Journal of Computer and System Sciences},
  volume={65},
  number={4},
  pages={672--694},
  year={2002},
  publisher={Elsevier}
}

@inproceedings{IW97,
  title={P= BPP if E requires exponential circuits: Derandomizing the XOR lemma},
  author={Impagliazzo, Russell and Wigderson, Avi},
  booktitle={Proceedings of the twenty-ninth annual ACM symposium on Theory of computing},
  pages={220--229},
  year={1997}
}

@article{NW94,
  title={Hardness vs.\ randomness},
  author={Nisan, Noam and Wigderson, Avi},
  journal={Journal of Computer and System Sciences},
  volume={49},
  number={2},
  pages={149--167},
  year={1994},
  publisher={Elsevier}
}

@article{BT10,
author = {Bravyi, Sergey and Terhal, Barbara},
title = {Complexity of Stoquastic Frustration-Free Hamiltonians},
journal = {SIAM Journal on Computing},
volume = {39},
number = {4},
pages = {1462-1485},
year = {2010},
doi = {10.1137/08072689X},
URL = {https://doi.org/10.1137/08072689X},
eprint = {https://doi.org/10.1137/08072689X}
}

@article{aharonov2013guest,
  title={Guest column: the quantum PCP conjecture},
  author={Aharonov, Dorit and Arad, Itai and Vidick, Thomas},
  journal={Acm sigact news},
  volume={44},
  number={2},
  pages={47--79},
  year={2013},
  publisher={ACM New York, NY, USA}
}

@article{aharonov2025stoqma,
  title={StoqMA vs. MA: the power of error reduction},
  author={Aharonov, Dorit and Grilo, Alex B and Liu, Yupan},
  journal={Quantum},
  volume={9},
  pages={1853},
  year={2025},
  publisher={Verein zur F{\"o}rderung des Open Access Publizierens in den Quantenwissenschaften}
}

@article{klassen2019two,
  title={Two-local qubit Hamiltonians: when are they stoquastic?},
  author={Klassen, Joel and Terhal, Barbara M},
  journal={Quantum},
  volume={3},
  pages={139},
  year={2019},
  publisher={Verein zur F{\"o}rderung des Open Access Publizierens in den Quantenwissenschaften}
}

@article{bravyi2006complexity,
  title={The complexity of stoquastic local Hamiltonian problems},
  author={Bravyi, Sergey and DiVincenzo, David P and Oliveira, Roberto and Terhal, Barbara M},
  journal={Quantum Information \& Computation},
  volume={8},
  number={5},
  pages={361--385},
  year={2008},
  publisher={Rinton Press}
}

@article{piddock2015complexity,
  title={The complexity of antiferromagnetic interactions and 2D lattices},
  author={Piddock, Stephen and Montanaro, Ashley},
  journal={Quantum Information \& Computation},
  volume={17},
  number={7-8},
  pages={636--672},
  year={2017},
  publisher={Rinton Press}
}

@article{bravyi2017complexity,
  title={On complexity of the quantum Ising model},
  author={Bravyi, Sergey and Hastings, Matthew},
  journal={Communications in Mathematical Physics},
  volume={349},
  number={1},
  pages={1--45},
  year={2017},
  publisher={Springer}
}

@article{gharibian2014quantum,
  title={Quantum hamiltonian complexity},
  author={Gharibian, Sevag and Huang, Yichen and Landau, Zeph and Shin, Seung Woo},
  journal={Foundations and Trends{\textregistered} in Theoretical Computer Science},
  volume={10},
  number={3},
  pages={159--282},
  year={2014},
  publisher={Now Publishers}
}

@article{cubitt2016complexity,
  title={Complexity classification of local Hamiltonian problems},
  author={Cubitt, Toby and Montanaro, Ashley},
  journal={SIAM Journal on Computing},
  volume={45},
  number={2},
  pages={268--316},
  year={2016},
  publisher={SIAM}
}

@article{bausch2017complexity,
  title={The complexity of translationally invariant spin chains with low local dimension},
  author={Bausch, Johannes and Cubitt, Toby and Ozols, Maris},
  journal={Annales Henri Poincar{\'e}},
  volume={18},
  number={11},
  pages={3449--3513},
  year={2017},
  organization={Springer}
}

@article{hastings2012trivial,
  title={Trivial low energy states for commuting Hamiltonians, and the quantum PCP conjecture},
  author={Hastings, Matthew B},
  journal={Quantum Information \& Computation},
  volume={13},
  number={5-6},
  pages={393--429},
  year={2013},
  publisher={Rinton Press}
}

@article{bravyi2003commutative,
  title={Commutative version of the k-local Hamiltonian problem and common eigenspace problem},
  author={Bravyi, Sergey and Vyalyi, Mikhail},
  journal={Quantum Information \& Computation},
  volume={5},
  number={3},
  pages={187--215},
  year={2005},
  publisher={Rinton Press}
}

@article{oliveira2005complexity,
  title={The complexity of quantum spin systems on a two-dimensional square lattice},
  author={Oliveira, Roberto and Terhal, Barbara M},
  journal={Quantum Information \& Computation},
  volume={8},
  number={10},
  pages={900--924},
  year={2008},
  publisher={Rinton Press}
}

@article{AGIK09-quantum-on-line,
  title={The power of quantum systems on a line},
  author={Aharonov, Dorit and Gottesman, Daniel and Irani, Sandy and Kempe, Julia},
  journal={Communications in mathematical physics},
  volume={287},
  number={1},
  pages={41--65},
  year={2009},
  publisher={Springer}
}

@article{REG26-geometric-completeness-stoqlh,
  title={Complexity of geometrically local stoquastic Hamiltonians},
  author={Raza, Asad and Eisert, Jens and Grilo, Alex B},
  journal={Quantum},
  volume={10},
  pages={2004},
  year={2026},
  publisher={Verein zur F{\"o}rderung des Open Access Publizierens in den Quantenwissenschaften}
}

@book{KSV02,
  title={Classical and quantum computation},
  author={Kitaev, Alexei Yu and Shen, Alexander and Vyalyi, Mikhail N},
  number={47},
  year={2002},
  publisher={American Mathematical Soc.}
}

\end{document}